\pdfoutput=1
\RequirePackage{ifpdf}
\ifpdf 
\documentclass[pdftex]{sigma}
\else
\documentclass{sigma}
\fi

\usepackage{mathtools}
\usepackage{euscript}

\newcommand{\ii}{{\rm i}}

\newcommand{\E}{\mathbb E}

\newtheorem{Theorem}{Theorem}[section]
\newtheorem{Corollary}[Theorem]{Corollary}
 { \theoremstyle{definition}
\newtheorem{Remark}[Theorem]{Remark} }

\newcommand{\ai}{\mathfrak a}
\newcommand{\Z}{\mathcal Z}

\begin{document}


\renewcommand{\thefootnote}{$\star$}

\newcommand{\arXivNumber}{1608.01557}

\renewcommand{\PaperNumber}{102}

\FirstPageHeading

\ShortArticleName{Moments Match between the KPZ Equation and the Airy Point Process}

\ArticleName{Moments Match between the KPZ Equation\\ and the Airy Point Process\footnote{This paper is a~contribution to the Special Issue on Asymptotics and Universality in Random Matrices, Random Growth Processes, Integrable Systems and Statistical Physics in honor of Percy Deift and Craig Tracy. The full collection is available at \href{http://www.emis.de/journals/SIGMA/Deift-Tracy.html}{http://www.emis.de/journals/SIGMA/Deift-Tracy.html}}}

\Author{Alexei BORODIN~$^{\dag\ddag}$ and Vadim GORIN~$^{\dag\ddag}$}

\AuthorNameForHeading{A.~Borodin and V.~Gorin}

\Address{$^\dag$~Department of Mathematics, Massachusetts Institute of Technology, USA}
\EmailD{\href{mailto:borodin@math.mit.edu}{borodin@math.mit.edu}, \href{mailto:vadicgor@gmail.com}{vadicgor@gmail.com}}

\Address{$^\ddag$~Institute for Information Transmission Problems of Russian Academy of Sciences, Russia}

\ArticleDates{Received August 09, 2016, in f\/inal form October 21, 2016; Published online October 26, 2016}

\Abstract{The results of Amir--Corwin--Quastel, Calabrese--Le~Doussal--Rosso, Dotsenko, and Sasamoto--Spohn imply that the one-point distribution of the solution of the KPZ equation with the narrow wedge initial condition coincides with that for a~multiplicative statistics of the Airy determinantal random point process. Taking Taylor coef\/f\/icients of the two sides yields moment identities. We provide a simple direct proof of those via a~combinatorial match of their multivariate integral representations.}

\Keywords{KPZ equation; Airy point process}

\Classification{60B20; 60H15; 33C10}

\renewcommand{\thefootnote}{\arabic{footnote}}
\setcounter{footnote}{0}

\section{Introduction}

Since Tracy and Widom's discovery of the ASEP solvability eight years ago \cite{TW2,TW1, TW3}, the relationship between the ``determinantal'' and ``non-determinantal'' solvable models in the \mbox{$(1+1)$}-dimensional KPZ (Kardar--Parisi--Zhang) universality class has largely remained a mystery. One step towards solving this mystery is the celebrated result of Amir--Corwin--Quastel \cite{ACQ}, Calabrese--Le~Doussal--Rosso \cite{CDR}, Dotsenko \cite{D}, and Sasamoto--Spohn \cite{SS}, that provides an explicit expression for the distribution (or its Laplace transform) of one-point value of the solution of the KPZ equation with the so-called narrow wedge initial condition. It can be re-interpreted as saying that this Laplace transform coincides with the average of a multiplicative statistics of the Airy determinantal random point process. Although this restatement seems to be known to experts, we couldn't f\/ind it in this form in the literature, so we give an exact formulation as Theorem~\ref{Theorem_Laplace} below. Such a result is very useful as it immediately implies that this solution of the KPZ equation asymptotically at large times has the GUE Tracy--Widom distribution, which is a display of the KPZ universality, cf.\ Corwin's survey~\cite{Ivan_review}.

Finding other facts of similar nature has been a challenge so far.

Imamura and Sasamoto \cite{IS} proved a similar statement for the O'Connell--Yor semi-discrete Brownian directed polymer. Unfortunately, the associated determinantal point process was not governed by a positive measure. Still, taking the edge limit of this process, they were able to recover Theorem~\ref{Theorem_Laplace}. Another representation of the Laplace transform of the O'Connell--Yor partition function as the average of a multiplicative functional over a signed determinantal point process can be found in~\cite{OC}.

Very recently, one of the authors found in~\cite{B} an identity that relates a single point height distribution of the (higher spin inhomogeneous) stochastic six vertex model in a quadrant on one side, and multiplicative statistics of the Macdonald measures on the other. The ASEP limit of this identity was worked out in~\cite{BO}. Taking the KPZ limit of both leads to Theorem~\ref{Theorem_Laplace} again.

The goal of this note is to look at Theorem~\ref{Theorem_Laplace} from the point of view of moments, rather than the corresponding distributions. One can study both the KPZ equation and the Airy point process via their exponential moments. Those are computationally tractable but they are of limited mathematical use because the corresponding moment problems are indeterminate. Still, on the KPZ side physicists were able to consistently use the moments to access the distributions via the (non-rigorous) replica trick, see~\cite{D} and~\cite{CDR} for early examples.

Our Theorem \ref{Theorem_moments} proves the moments identity that corresponds to Theorem~\ref{Theorem_Laplace}. The argument is a combinatorial match between known multivariate integral representations of the moments on both sides. Interestingly, these integral representations were known long before~\cite{ACQ,CDR,D, SS}, but their similarity had not been exploited. We are hoping that the moments point of view will be benef\/icial for f\/inding other similar correspondence.

We did attempt to extend the moments correspondence to a two-point identity, as integral representations on both side are again known. Unfortunately, we have not been successful in that so far.

\section{The one-point equality}

Let $\ai_1\ge\ai_2\ge\ai_3\ge\cdots$ be points of the Airy point process\footnote{It should not be confused with Airy$_2$ process; the latter is a random continuous curve rather than a random point process.} at $\beta=2$ (see, e.g., \cite{AGZ,Forrester}) which is a determinantal point process on $\mathbb R$ with correlation kernel
\begin{gather*}
 K_{\rm Airy}(x,y)=\frac{\operatorname{Ai}(x) \operatorname{Ai}'(y)-\operatorname{Ai}'(x) \operatorname{Ai}(y)}{x-y} = \int_{0}^\infty \operatorname{Ai}(x+a) \operatorname{Ai}(y+a) da.
\end{gather*}
Here $\operatorname{Ai}(x)$ is the Airy function.

From the opposite direction, let $\Z(T,X)$ denote the solution of the stochastic heat equation (see, e.g., \cite{Ivan_review, Quastel})
\begin{gather*}
 \frac{\partial}{\partial T} \Z=\frac{1}{2} \frac{\partial^2}{\partial X^2} \Z - \Z \dot{\mathcal W},\qquad \Z(0,X)=\delta(X=0),
\end{gather*}
where $\dot{\mathcal W}$ is the space--time white noise. $H:=-\log(\Z)$ is the Hopf--Cole solution of the Kardar--Parisi--Zhang stochastic partial dif\/ferential equation with the narrow-edge initial data.

The following statement is a reformulation of results of \cite{ACQ, CDR, D, SS}.
\begin{Theorem} \label{Theorem_Laplace} Set $\frac{T}{2}=C^3$. Then for each real $C>0$, $u\ge 0$ we have
\begin{gather}\label{eq_Laplace}
 \E_{\mathrm{Airy}}\left[ \prod_{k=1}^{\infty} \frac{1}{1+u\exp\left(
 C\ai_k\right)}\right]=\E_{\mathrm{KPZ}}\left[ \exp\left({-}u \Z(T,0) \exp\left(\frac{T}{24}\right)\right)\right].
\end{gather}
\end{Theorem}

On the other hand, we could not f\/ind the following statement in the existing literature.

\begin{Theorem} \label{Theorem_moments} Set $\frac{T}{2}=C^3$, and let $h_k(x_1,x_2,\dots)=\sum\limits_{i_1\le i_2\le \dots \le i_k} x_{i_1} x_{i_2} \cdots x_{i_k}$ be the complete symmetric homogeneous function in variables $x_1,x_2,\dots$. Then for each $C>0$, $k=1,2,\dots$ we have
\begin{gather}\label{eq_moments}
 \E_{\mathrm{Airy}}\left[ h_k\left(\exp\left(
 C\ai_1\right),\exp\left( C\ai_2\right),\dots \right)\right]=\E_{\mathrm{KPZ}}\left[
 \frac{\Z(T,0)^k}{k!}\right] \exp\left(k \frac{T}{24}\right).
\end{gather}
\end{Theorem}

Our proof of Theorem \ref{Theorem_moments} is based on direct comparison of contour integral formulas: for the right-hand side of~\eqref{eq_moments} such formula is known as a solution for the attractive delta Bose gas equation, cf.\ the discussion in \cite[Section~6.2]{BigMac}, while for the left-hand side it can be computed through the Laplace transform of the correlation kernel $K_{\rm Airy}$.

\begin{Remark} Expanding formally the result of Theorem \ref{Theorem_Laplace} into power series in $u$ and evaluating the coef\/f\/icients, one gets the result of Theorem~\ref{Theorem_moments} and vice versa. However, these theorems are not equivalent: The $u$ power series expansion of the left-hand side of~\eqref{Theorem_Laplace} fails to converge for any $u\ne 0$. Below we provide two dif\/ferent proofs for Theorems~\ref{Theorem_Laplace} and~\ref{Theorem_moments}, respectively.
\end{Remark}

The following corollary is present in~\cite{ACQ, CDR, D, SS}, but it seems reasonable for us to give a~proof using Theorem~\ref{Theorem_Laplace} above only, without appealing to the explicit evaluation of either side.

\begin{Corollary} The following convergence in distribution holds:
\begin{gather*} 
 \lim_{T\to +\infty}\left[ \left(\frac{2}{T}\right)^{1/3} \left( \ln(Z(T,0))+ \frac{T}{24}\right)\right]=\ai_1.
\end{gather*}
\end{Corollary}
\begin{proof}
Take $a\in\mathbb R$ and set $u=\exp\bigl({-}(T/2)^{1/3} a\bigr)$. Then the left-hand side of \eqref{eq_Laplace} is
\begin{gather} \label{eq_x0}
 \E\left[ \prod_{k=1}^{\infty} \frac{1}{1+\exp\bigl( (T/2)^{1/3} (\ai_k-a)\bigr)}\right].
\end{gather}
We claim that the random variable under expectation in~\eqref{eq_x0} almost surely converges as $T\to\infty$ to the indicator function of the event $\ai_1<a$. Indeed, if $\ai_1>a$, then the expression in \eqref{eq_x0} converges to $0$ as $T\to\infty$. On the other hand, if $\ai_1<a$, then (taking into the account that $\sum\limits_{k=1}^{\infty} \exp(\ai_k)$ is almost surely f\/inite, as follows from the f\/initeness of its expectation, which is explicitly known, see, e.g., \cite[Section~2.6.1]{Ok2}) the same expression converges to $1$. Since the random variable under expectation is almost surely between $0$ and $1$, the almost sure convergence implies convergence of expectations, and therefore
\begin{gather} \label{eq_x1}
 \lim_{T\to +\infty} \E\left[ \prod_{k=1}^{\infty} \frac{1}{1+\exp\bigl( (T/2)^{1/3} (\ai_k-a)\bigr)}\right]={\rm Prob}(\ai_1<a).
\end{gather}
On the other hand, the right side of \eqref{eq_Laplace} is
\begin{gather} \label{eq_x2}
\E\left[ \exp\left(-\exp\left( (T/2)^{1/3} \left( \frac{\ln(\Z(T,0)) +T/24}{(T/2)^{1/3}} -a \right)\right)\right)\right].
\end{gather}
Observe that for any random variable $\xi$, the expression $\exp\bigl({-}\exp((T/2)^{1/3}(\xi-a))\bigr)$ is almost surely between $0$ and $1$, converges to $0$ as $T\to +\infty$ if $\xi>a$, and converges to $1$ if $\xi<a$. Therefore, using the fact that $\ai_1$ has a continuous distribution, \eqref{eq_x2} as $T\to\infty$ behaves as (see, e.g., \cite[Lemma~4.1.39]{BigMac} for more details)
\begin{gather} \label{eq_x3}
 {\rm Prob} \left(\frac{\ln(\Z(T,0)) +T/24}{(T/2)^{1/3}}<a\right)+o(1).
\end{gather}
Equating \eqref{eq_x1} to \eqref{eq_x3} we are done.
\end{proof}

\section{Proofs of Theorems \ref{Theorem_Laplace} and \ref{Theorem_moments}}

\begin{proof}[Proof of Theorem \ref{Theorem_Laplace}]
The expectation $\E \exp(-u \Z(T,0)$ admits a formula as a Fredholhm determinant, which was discovered in \cite{ACQ, CDR, D, SS}. Following \cite[Section~2.2.1]{BBC}, this formula reads
\begin{gather}\label{eq_KPZ_Laplace}
\E\big[ \exp\bigl({-}u \Z(T,0) \exp(T/24)\bigr) \big]=1+\sum_{L=1}^{\infty} \frac{(-1)^L}{L!}
 \int\limits_0^{\infty}dx_1\cdots \int\limits_0^{\infty}dx_L \det\left[ K_u(x_i,x_j)\right]^{L}_{i,j=1},
\end{gather}
where
\begin{gather*}
 K_u(x,x')=\int_{-\infty}^{\infty} \frac{dr}{1+\frac{1}{u}\exp\big( (T/2)^{1/3} r\big)} \operatorname{Ai}(x-r) \operatorname{Ai}(x'-r).
\end{gather*}
On the other hand, the left-hand side of \eqref{eq_Laplace} is a multiplicative function of a determinantal point process and, therefore, also admits a~Fredholm determinant formula, see, e.g., \cite[equation~(2.4)]{Bor_det}:
\begin{gather}
 \E\left[ \prod_{k=1}^{\infty} \frac{1}{1+u\exp\left(
 C\ai_k\right)}\right]=\det\left[1- \left(1- \frac{1}{1+u\exp(Cr)}\right)
 K_{\rm Airy}(r,r')\right]_{L_2(\mathbb R)}\nonumber\\
 =
1+\sum_{L=1}^{\infty} \frac{(-1)^L}{L!} \int_{-\infty}^{\infty} dy_1\cdots
\int_{-\infty}^{\infty} d y_L \left(\prod_{k=1}^L \frac{1}{1+\frac{1}{u} \exp(-C
y_k)} \right) \det\left[K_{\rm Airy}(y_i,y_j)\right]_{i,j=1}^L.\!\!\!\label{eq_Airy_multiplicative}
\end{gather}
One immediately sees that upon the change of variables $r_i=-y_i$ and identif\/ication $\frac{T}{2}=C^3$, the formulas~\eqref{eq_KPZ_Laplace} and~\eqref{eq_Airy_multiplicative} are the same.
\end{proof}

\begin{proof}[Proof of Theorem \ref{Theorem_moments}] The moments $\E \Z(T,0)^k$ are known through solving the attractive delta Bose gas equation, cf.\ the discussion in \cite[Section~6.2]{BigMac}. Following \cite[Lemma~4.1]{BBC}, we have
\begin{gather}\label{eq_Bose_moment}
 \E \big[\Z(T,0)^k\big] =\int_{a_1-\ii\infty}^{a_1+\ii\infty}\frac{dz_1}{2\pi\ii} \cdots
 \int_{a_k-\ii\infty}^{a_k+\ii \infty} \frac{dz_k}{2\pi\ii} \prod_{1\le A < B \le k}
 \frac{z_A-z_B}{z_A-z_B-1} \cdot \prod_{j=1}^k \exp\left( \frac{T}{2} z_j^2\right),
\end{gather}
where the real numbers $a_1,\dots,a_k$ satisfy $a_1\gg a_2\gg\cdots\gg a_k$. It is convenient for us to modify the contours of integration in~\eqref{eq_Bose_moment} to the imaginary axis $\ii \mathbb R$. One collects certain residues in such a~deformation, and the f\/inal result is read from \cite[equation~(13)]{BBC} to be
\begin{gather}
\E \left[\frac{\Z(T,0)^k}{k!}\right] = \sum_
 {\begin{smallmatrix}\lambda\vdash k\\ \lambda=1^{m_1} 2^{m_2}\cdots
 \end{smallmatrix}} \frac{1}{m_1! m_2! \cdots}
 \int_{-\ii \infty}^{\ii \infty} \frac{dw_1}{2\pi\ii} \cdots \int_{-\ii\infty}^{\ii\infty} \frac{dw_{\ell(\lambda)}}{2\pi\ii}
 \det\left[ \frac{1}{w_j+\lambda_j-w_i}\right]_{i,j=1}^{\ell(\lambda)}\nonumber \\
\hphantom{\E \left[\frac{\Z(T,0)^k}{k!}\right] = }{}
 \times \prod_{j=1}^{\ell(\lambda)} \exp\left( \frac{T}{2} \big(w_j^2 +(w_j+1)^2+\cdots+(w_j+\lambda_j-1)^2 \big)\right),\label{eq_Bose_moment_expanded}
\end{gather}
where $\lambda=(\lambda_1\ge\lambda_2\ge\dots)$ is a partition of $k$ and $\ell(\lambda)$ is the number of non-zero parts $\lambda_j$.

Let us now produce a similar expression for the left-hand side of~\eqref{eq_moments}. Def\/ine the Laplace transform of the correlation functions of the Airy point process through
\begin{gather*}
 R(c_1,\dots,c_n)=\int_{\mathbb R^n} e^{(c \cdot x)} \det [K_{\rm Airy}(x_i,x_j)]_{i,j=1}^n dx_1\cdots dx_n, \qquad c_1,c_2,\dots,c_n>0.
\end{gather*}
The def\/inition of the Airy point process implies that for a partition $\lambda=(\lambda_1\ge\lambda_2\dots\ge\lambda_\ell)=1^{m_1} 2^{m_2}\cdots$ one has
\begin{gather*}
 \E \bigl[m_\lambda(\exp(C\ai_1),\exp(C\ai_2),\dots)\bigr]=\frac{1}{m_1!m_2!\cdots} R(C\lambda_1,\dots,C\lambda_\ell),
\end{gather*}
where $m_\lambda(y_1,y_2,\dots)$ is the monomial symmetric function in variables $y_1,y_2,\dots$, as in \cite[Chapter~I]{Mac}. Expanding $h_k$ into linear
combination of $m_\lambda$'s, we can then write
\begin{gather}\label{eq_h_expansion}
 \E \bigl[h_k(\exp(C\ai_1),\exp(C\ai_2),\dots)\bigr]=\sum_{\begin{smallmatrix}\lambda\vdash k
 \\\lambda=1^{m_1} 2^{m_2}\cdots \end{smallmatrix}} \frac{1}{m_1!m_2!\cdots} R\bigl(C\lambda_1,\dots,C\lambda_{\ell(\lambda)}\bigr),
\end{gather}
where the summation goes over all partitions of $k$. Comparing \eqref{eq_h_expansion} with~\eqref{eq_Bose_moment_expanded}, we see that it remains to identify the contour integrals over imaginary axis in~\eqref{eq_Bose_moment_expanded} with $R\bigl(C\lambda_1,\dots,C\lambda_{\ell(\lambda)}\bigr)$.

The rest is based on the following identity that can be found in \cite[Lemma~2.6]{Ok2}:
\begin{gather*}
 \int_{-\infty}^{+\infty} e^{xz} \operatorname{Ai}(z+a) \operatorname{Ai}(z+b) dz =\frac{1}{2\sqrt{\pi x}} \exp\left( \frac{x^3}{12} - \frac{a+b}{2}
x - \frac{(a-b)^2}{4x} \right), \qquad x>0.
\end{gather*}
Its immediate corollary is (we use an agreement $z_{n+1}=z_1$ and $s_{n+1}=s_1$ here, and also assume $c_1,\dots,c_n>0$)
\begin{gather*}
 {\mathcal E}(c_1,\dots,c_n):=\int_{\mathbb R^n} e^{(c\cdot z)} \prod_{i=1}^n K_{\rm Airy}(z_i,z_{i+1}) dz= \frac{1}{2^n
\pi^{n/2}} \frac{e^{\sum c_i^3/12}}{\prod\limits_{i=1}^n \sqrt{c_i}}\\
\hphantom{{\mathcal E}(c_1,\dots,c_n):=}{}
\times \int_{s_1\ge 0} \cdots \int_{s_n\ge 0} \exp\left( - \sum_{i=1}^n \frac{(s_i-s_{i+1})^2}{4 c_i} -\sum_{i=1}^n \frac{s_i+s_{i+1}}{2} c_i \right)
\prod_{i=1}^n ds_i.
\end{gather*}
Using the Gaussian integrals in variables $z_1,\dots,z_n$, the last formula is converted into
\begin{gather}
 {\mathcal E}(c_1,\dots,c_n)= \frac{1}{(2\pi)^n} e^{\sum c_i^3/12} \int_{s_1\ge 0} ds_1 \cdots \int_{s_n\ge 0} ds_n
 \int_{z_1\in\mathbb R} d z_1 \cdots \int_{z_n\in\mathbb R} d z_n \nonumber\\
 \hphantom{{\mathcal E}(c_1,\dots,c_n)=}{} \times \exp\left( \sum_{i=1}^n\big( {-}c_i z_i^2 +\ii
(z_i-z_{i+1})s_i - (c_i+c_{i+1})s_i/2 \big)\right).\label{eq_x4}
\end{gather}
Since $\ii (z_i - z_{i+1})$ has zero real part, we can integrate over $s_i$ in~\eqref{eq_x4}, arriving at the formula:
\begin{gather}
 {\mathcal E}(c_1,\dots,c_n) = \frac{e^{\sum x_i^3/12}}{(2\pi)^n
} \int_{z_1\in\mathbb R} \cdots \int_{z_n\in\mathbb R} \exp\left( -\sum_{i=1}^n
c_i z_i^2 \right) \prod_{i=1}^n \frac{dz_i}{-\ii (z_i-z_{i+1}) + \frac{c_i+c_{i+1}}{2}}.\label{eq_x35}
\end{gather}
We can now write the formula for $R(c_1,\dots,c_n)$ (we subdivide a permutation into cycles, use~\eqref{eq_x35} and then combine back):
\begin{gather*}
R(c_1,\dots,c_n)=\int_{\mathbb R^n} d x_1\cdots dx_n \sum_{\sigma\in {\mathfrak
S}(n)} (-1)^{\sigma} \prod_{j=1}^n e^{x_j c_j} K_{\rm Airy}(x_j,x_{\sigma(j)})\nonumber\\
=\frac{e^{\sum c_i^3/12}}{(2\pi)^n } \int_{z_1\in\mathbb R} \cdots
\int_{z_n\in\mathbb R} \exp\left( -\sum_{i=1}^n c_i z_i^2 \right) \sum_{\sigma\in
{\mathfrak S}(n)} (-1)^{\sigma} \prod_{i=1}^n \frac{dz_i}{-\ii (z_i-z_{\sigma(i)}) +\frac{c_i+c_{\sigma(i)}}{2}}\nonumber\\
=\frac{e^{\sum c_i^3/12}}{(2\pi)^n } \int_{z_1\in\mathbb R} dz_1 \cdots
\int_{z_n\in\mathbb R} dz_n \exp\left( -\sum_{i=1}^n c_i z_i^2 \right)
 \det\left[\frac{1}{ \big({-}\ii z_i + \frac{c_i}{2}\big)+\big(\ii z_j+\frac{c_j}{2})}\right]_{i,j=1}^n .
\end{gather*}
\begin{Remark} \label{Remark_Cauchy} We can use the Cauchy determinant formula
\begin{gather*}
 \det\left[\frac{1}{a_i+b_j}\right]_{i,j=1}^n=\prod_{i=1}^n \frac{1}{a_i+b_i} \prod_{1\le i <j \le n}\frac{(a_i-a_j)(b_i-b_j)}{(a_i+b_j) (a_j+b_i)}
\end{gather*}
with $a_i=-\ii z_i + c_i/2$, $b_i=\ii z_i+c_i/2$ to simplify the last determinant.
\end{Remark}

We now take a partition $\lambda \vdash k$ with $\ell(\lambda)=n$, set $c_i=C\lambda_i$ and make a change of variables
\begin{gather*}
 \ii z_j= C w_j + \frac{C\lambda_j}{2}-\frac{C}{2}
\end{gather*}
to get (note that we deformed the contours to the imaginary axis; we do not pick up any residues in such a deformation)
\begin{gather}
 R(C\lambda_1,\dots,C\lambda_n)= \frac{\exp\Big( C^3 \sum\limits_{i=1}^n \lambda_i^3/12\Big)}{(2\pi \ii)^n } \int_{-\ii \infty}^{\ii \infty} dw_1 \cdots
\int_{-\ii \infty}^{\ii \infty} dw_n \nonumber \\
\hphantom{R(C\lambda_1,\dots,C\lambda_n)=}{} \times \exp\left( C^3 \sum_{i=1}^n
\lambda_i (w_i+\lambda_i/2-1/2)^2 \right) \det\left[\frac{1}{ w_j+\lambda_j - w_i}\right]_{i,j=1}^n.\label{eq_R_formula}
\end{gather}
It remains to simplify the exponents:
\begin{gather}
 \exp\left[ C^3 \sum_{i=1}^n \left(\frac{\lambda_i^3}{12} +\lambda_i\left(w_i+\frac{\lambda_i}2-\frac12\right)^2\right)\right]\nonumber\\
 \qquad {}= \exp\left[ C^3 \sum_{i=1}^n \left( \lambda_i w_i^2 +\frac{\lambda_i^3}{3}
 +\frac{\lambda_i}{4} +w_i\lambda_i^2 - w_i \lambda_i-\frac{\lambda_i^2}{2} \right)\right]\nonumber\\
\qquad{}=
 \exp\left[ C^3 \sum_{i=1}^n \left( \lambda_i w_i^2 + \lambda_i(\lambda_i-1) w_i +\frac{\lambda_i(\lambda_i-1)(2\lambda_i-1)}{6}
 -\frac{\lambda_i}{12} \right)\right]\nonumber\\
 \qquad {}= \exp\left[ C^3 \sum_{i=1}^n \big( w_i^2+(w_i+1)^2+\dots+(w_i+\lambda_i-1)^2 \big)-C^3 \frac{k}{12}\right]. \label{eq_exponents_simple}
\end{gather}
Combining \eqref{eq_h_expansion} with \eqref{eq_R_formula}, \eqref{eq_exponents_simple} and identifying $C^3=\frac{T}{2}$. we arrive at~\eqref{eq_Bose_moment_expanded} multiplied by $\exp(-kT/24)$.
\end{proof}

\subsection*{Acknowledgements}

A.B.~was partially supported by the NSF grants DMS-1056390 and DMS-1607901. V.G.~was partially supported by the NSF grant DMS-1407562 and by the Sloan Research Fellowship.

\pdfbookmark[1]{References}{ref}
\LastPageEnding

\end{document}